\documentclass[pra,a4paper,aps,twocolumn,superscriptaddress,longbibliography]{revtex4-2}

\usepackage{amsmath,amsthm,amssymb,amsfonts,setspace}
\usepackage{bm}
\usepackage{float}
\usepackage{epsfig,color}

\usepackage[usenames,dvipsnames]{xcolor}
\usepackage[colorlinks]{hyperref}
\hypersetup{
    allcolors  = {blue},
}
\usepackage{mathtools}
\usepackage{cleveref} 

\usepackage{soul}

\usepackage[normalem]{ulem}
\usepackage{graphicx}

\crefformat{equation}{Eq.~(#2#1#3)} 
\crefformat{section}{Section~#2#1#3}
\crefformat{theorem}{Theorem~#2#1#3}
\crefformat{lemma}{Lemma~#2#1#3}
\crefformat{appendix}{Appendix~#2#1#3}
\Crefformat{equation}{Equation~(#2#1#3)}
\crefformat{figure}{Fig.~#2#1#3}
\crefformat{corollary}{Corollary~#2#1#3}
\crefrangeformat{equation}{Eqs.~#3(#1)#4--#5(#2)#6}

 \theoremstyle{plain}
 
 \theoremstyle{plain}
 
 \theoremstyle{plain}

 \theoremstyle{plain}
 
 \theoremstyle{plain}
 
 \theoremstyle{plain}
 
 \theoremstyle{remark}
 \newtheorem*{rem*}{Remark}
 \theoremstyle{plain}

\theoremstyle{plain}

\newcommand{\ket}[1]{\vert{#1}\rangle} 

 \newcommand{\ketbra}[2]{\left\vert{#1}\right\rangle\!\left\langle{#2}\right\vert}

\newcommand{\Tr}{\mathrm{Tr}}

\theoremstyle{plain}
\newtheorem{theorem}{Theorem}

\theoremstyle{plain}
\newtheorem{corollary}{Corollary} 
\theoremstyle{plain}

\theoremstyle{plain}

\theoremstyle{plain}

\theoremstyle{plain}

\begin{document}

\title{{A simple proof that a}nomalous weak values require coherence}

\newcommand{\inl}{INL -- International Iberian Nanotechnology Laboratory, Av. Mestre Jos\'{e} Veiga s/n, 4715-330 Braga, Portugal}
\newcommand{\inlshort}{INL -- International Iberian Nanotechnology Laboratory, Braga, Portugal}
\newcommand{\uff}{Instituto de F\'{i}sica, Universidade Federal Fluminense, Av. Gal. Milton Tavares de Souza s/n, Niter\'{o}i -- RJ, 24210-340, Brazil}
\newcommand{\uffshort}{Instituto de F\'{i}sica, Universidade Federal Fluminense, Niter\'{o}i -- RJ, Brazil}
\newcommand{\cfum}{Centro de F\'{i}sica, Universidade do Minho, Campus de Gualtar, 4710-057 Braga, Portugal}
\newcommand{\cfumshort}{Centro de F\'{i}sica, Universidade do Minho, Braga, Portugal}

\author{Rafael Wagner}
\email{rafael.wagner@inl.int}
\affiliation{\inlshort}
\affiliation{\cfumshort}

\author{Ernesto F. Galv\~ao}
\email{ernesto.galvao@inl.int}
\affiliation{\inlshort}
\affiliation{\uffshort}

\date{\today}

\begin{abstract}
The quantum mechanical weak value $A_w=\left\langle {\phi}|A|\psi \right \rangle / \left\langle \phi | \psi \right\rangle$ of an observable $A$ is a measurable quantity associated with an observable $A$ and pre- and post-selected states $\vert\psi \rangle, \vert \phi \rangle$. Much has been discussed about the meaning and metrological uses of anomalous weak values, lying outside of the range of eigenvalues of $A$. 
We { present a simple proof} that anomalous weak values require that the (possibly mixed) pre- and post- selection states have coherence in the eigenbasis of $A$. We also present conditions under which anomalous $A_w$ are witnesses of generalized contextuality, dispensing with the operational weak measurement set-up.

\end{abstract}

\maketitle

\section{Introduction}\label{sec: intro}
Superposition states are a defining hallmark of quantum mechanics. For general mixed states this resource is known as quantum coherence, and is defined with respect to a specific choice of basis $\{\vert a \rangle \}_{a}$ associated with a (non-degenerate) observable $A$. In this context, \textit{coherent states} $\rho$ are defined as those which have non-null off-diagonal density matrix elements $\langle a\vert \rho \vert a'\rangle \neq 0$ for $a \neq a' $.  Coherence can be formally treated as a resource~\cite{gour2019quantumresource}, and shown to be responsible for various nonclassical phenomena, providing advantage in information processing  tasks~\cite{streltsov2017colloquium}. 

While standard projective measurements typically strongly disturb a quantum system, in 1988 Aharonov, Albert and Vaidman~\cite{aharonov1988result} proposed a new measurement scheme allowing for a tunable degree of disturbance on the measured systems. A weak measurement scheme involves preparing a quantum state $\vert \psi \rangle$, followed by a weak interaction between the system and a measurement apparatus, generated by some observable $A$, with a final postselection onto some other state $\vert \phi \rangle$. The average change in the apparatus pointer, for a sufficiently weak interaction between the measurement device and the state, will be given by the so-called weak value
\begin{equation}
A_w = \frac{\langle \phi \vert A \vert \psi \rangle }{\langle \phi \vert \psi \rangle}.    
\end{equation}

The weak value $A_w$ differs from common averages of the observable $A$ in that it can lie outside the range of the spectrum $\sigma(A)$ of $A$. When this happens, $A_w$ is called an \textit{anomalous} weak value, and this property has been shown to provide some advantage in  metrology~\cite{dressel2014colloquium,tamir2013introduction}. It has been argued that classical interference models can reproduce this effect~\cite{ferrie2014how}. Later it was shown that those effects would only be possible for models capable of precisely reproducing the same kind of interference phenomenology that makes non-classical effects possible for physical systems \cite{dressel2015weak}. Under specific operational constraints, statistics arising from anomalous weak values in weak measurements was shown to be explained only by contextual models~\cite{pusey2014anomalous,kunjwal2019anomalous}.

We advance the analysis of the role of coherence in weak values~\cite{dressel2015weak,mundarain2016quantumness,pan2020interference,wagner2023quantum} by studying the quasi-probability distribution mentioned in Ref.~\cite{dressel2015weak}, revisiting it from the perspective of unitary-invariant properties of a set of quantum states known as Bargmann invariants \cite{chien2016characterization, oszmaniec2021measuring}. 
Our Theorem~\ref{theorem: main result} shows that weak value anomaly requires a rather specific type of coherence to be present, namely, coherence as a relational property between the pre/post selection states and the eigenbasis of the observable $A$. We provide several examples showing that coherence alone is not sufficient for anomaly to appear. In Corolary ~\ref{corollary: main resul} we also show that negativity or imaginarity of the quasi-probabilities guarantees anomalous weak values for certain observables.

The fact that weak value anomaly implies coherence opens up the possibility of witnessing coherence using weak value measurements, without the need for state tomography, knowledge of dimension, purity or commutativity. This could be done using recently proposed quantum circuits that  measure weak values \cite{wagner2023quantum, oszmaniec2021measuring}. We also remark on the relevance of recently established connections between unitary invariants and contextuality~\cite{wagner2022inequalities}, together with techniques for testing contextuality without relying on operational constraints~\cite{schmid2018contextual,selby2021accessible,selby2022open}. These results enable us here to present simple ways to robustly quantify contextuality using measurements of weak values, allowing  novel simplified tests of contextuality.

\section{A quasi-probability distribution associated with weak values.}\label{sec: quasiprobability}

Consider an arbitrary observable $A$, with eigenbasis $\{ \ket{a_i}\}_{i=1}^d$ and corresponding eigenvalues $\{a_i\}_{i=1}^d$. The weak value~\cite{aharonov1988result,aharonov1990properties} $A_w$ of $A$ is defined as:
\begin{equation}
    A_w = \frac{\langle \phi \vert A\vert \psi \rangle}{\langle \phi \vert \psi \rangle}= \sum_{i} a_i \frac{\langle \phi \vert a_i \rangle \langle a_i \vert \psi \rangle}{\langle \phi \vert \psi \rangle },
    \label{eq:awdef}
\end{equation}
where we assume $\langle \phi \vert \psi \rangle \neq 0$. In Refs.~\cite{wagner2023quantum,yunger_Halpern2018quasiprobability,dressel2015weak}, it was observed that multiplying by $\langle \psi | 
\phi \rangle / \langle \psi | \phi \rangle=1$ we can rewrite this as:
\begin{eqnarray}
    A_w = 
    \sum_i a_i \frac{\langle \phi \vert a_i \rangle \langle a_i \vert \psi \rangle \langle \psi \vert \phi \rangle }{\vert \langle \phi \vert \psi \rangle\vert^2 }
    = \sum_{i} a_i \frac{\Delta_3(\rho_{\phi},a_i,\rho_{\psi})}{\Delta_2(\rho_{\phi},\rho_{\psi})},
    \label{eq:awbarg}
\end{eqnarray}
where $\rho_{\psi} = \vert \psi \rangle \langle \psi \vert, \rho_{\phi} = \vert \phi \rangle \langle \phi \vert $ are, respectively, the pre- and post-selected states, and $\Delta_n$ is the $n$-th order Bargmann invariant \cite{bargmann1964note, chien2016characterization} of a $n$-tuple of states:
\begin{equation}\label{eq: bargmann invariant definition}
    \Delta_{\rho_1\rho_2\dots\rho_n}\equiv \Delta_n(\rho_1,\dots,\rho_n) = \text{Tr}(\rho_1 \dots \rho_n).
\end{equation}

Bargmann invariants are capable of witnessing the presence of a recently introduced notion of nonclassicality, termed \textit{set coherence}~\cite{designolle2021set}. It corresponds to the property that states in a given set cannot all be diagonal with respect to any single basis, as investigated in Ref.~\cite{galvao2020quantum}. Ref.~\cite{wagner2023quantum} noted that negativity and imaginarity of weak values are witnesses of set coherence. In particular, for the case of weak values of a given observable $A$, this will imply coherence with respect to the eigenbasis of $A$. Fig.~\ref{fig:awgraph} shows the graph characterizing relational information for all quantities defining $A_w$. Our treatment of coherence applies to general pre- and post-selection states, including mixed states, for which weak values are defined in terms of Bargmann invariants as
\begin{eqnarray}
    A_w &=& \frac{\text{Tr}(\rho_\phi A \rho_\psi)}{\text{Tr}(\rho_\phi\rho_\psi)} =  \sum_{i}a_i \frac{\text{Tr}(\rho_\phi \vert a_i \rangle \langle a_i \vert \rho_\psi)}{\text{Tr}(\rho_\phi \rho_\psi)}= \\ &=&\sum_{i}a_i \frac{\Delta_3(\rho_{\phi},a_i,\rho_{\psi})}{\Delta_2(\rho_{\phi},\rho_{\psi})} \label{eq:mixedb}
\end{eqnarray}

Circuits based on the cycle test \cite{oszmaniec2021measuring} were proposed in Ref.~\cite{wagner2023quantum} to directly estimate $A_w$ in this more general form, which has appeared elsewhere~\cite{dressel2015weak,dziewior2019universality}. It is easy to check that in Eqs. (\ref{eq:awbarg}) and (\ref{eq:mixedb}), the weight terms $\frac{\Delta_3(\rho_{\phi},a_i,\rho_{\psi})}{\Delta_2(\rho_{\phi},\rho_{\psi})}$ define quasi-probabilities, in the sense that these terms sum to 1: 
\begin{equation}\label{eq: quasi wv}
    g(\rho_\phi,\rho_\psi\vert a_i) := \frac{\Delta_3(\rho_{\phi},a_i,\rho_{\psi})}{\Delta_2(\rho_{\phi},\rho_{\psi})}, \sum_i g(\rho_\phi,\rho_\psi\vert a_i) = 1.
\end{equation}
However, these quasi-probabilities can be \textit{anomalous}, that is, outside of the real interval $[0,1]$.  The quasi-probabilities $g(\rho_\phi,\rho_\psi\vert a_i)$ characterize relational, unitary-invariant properties of the set of states that includes $A$'s eigenbasis and the two states $\rho_{\psi}, \rho_{\phi}$. In related recent work, negativity and imaginarity of the Kirkwood-Dirac (KD) quasi-probability distribution has been linked to anomalous weak values \cite{yunger_Halpern2018quasiprobability,lostaglio2022kirkwood}. As pointed out in \cite{wagner2023quantum}, the KD distribution is written in terms of relational properties of a single state and two different bases. By focusing on the minimal scenario involving just $A$'s eigenbasis and the pre- and post-selected states, we will obtain a sharper characterization of the connection between anomalous weak values and coherence. 

\begin{figure}               \includegraphics[width=0.95\columnwidth]{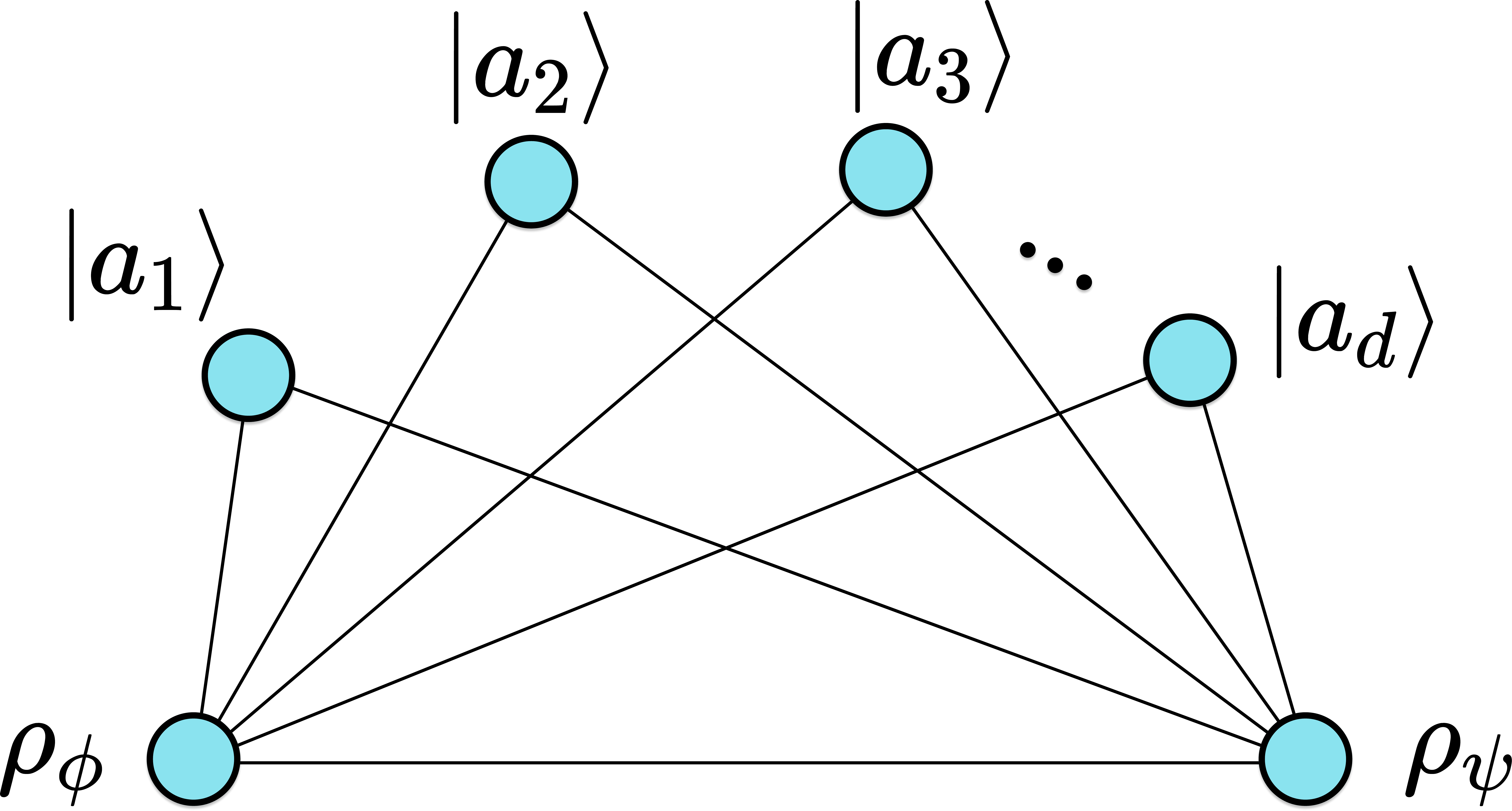}
\centering
    \caption{\textbf{Frame graph characterizing relational information of pre- and post-selection states, and a basis for observable $A$.} Vertices represent quantum states in a $d$-dimensional Hilbert space. 
    For general mixed states, edges represent two-state overlaps $\text{Tr}(\rho_\psi \rho_\phi)$. Two of the vertices represent the pre-selected state $\rho_{\psi}$ and the post-selected state $\rho_{\phi}$, with all other vertices representing the vector eigenbasis of $A$, the observable of interest.}
    \label{fig:awgraph}
\end{figure}

Constructions related to our proposed distribution $g$ have appeared before in the literature. In the continuous phase-space setting, a complex-valued quasi-probability distribution was introduced in the context of the cross-Wigner distribution~\cite{degosson2012weak}. 
The distribution $g$ can be viewed as a discrete version of the continuous distribution $\rho_{\phi, \psi}(z)$ described in Ref.~\cite{degosson2012weak}, sharing its key properties, but with respect to the discrete phase space provided by the eigenbasis of $A$. Denoting $g_i \equiv g(\rho_\phi,\rho_\psi | a_i)$, it is easy to show  that: (i) $\sum_i \text{Re}[g_i] = 1, \sum_i \text{Im}[g_i]=0 $, (ii) $g(\rho_\phi,\rho_\psi\vert a_i)^* = g(\rho_\psi,\rho_\phi \vert a_i)$, (iii) $g(\alpha \rho_\phi, \alpha \rho_\psi \vert a_i) = g(\rho_\phi,\rho_\psi \vert a_i), \forall \alpha \in \mathbb{C}$, and (iv) $A_w = \sum_i a_i g_i$. All these properties are also satisfied by $\rho_{\phi,\psi}(z)$ in Ref.~\cite{degosson2012weak} where $z = (x,p) \in \mathbb{R}^{2N}$ constitutes the continuous phase-space of $N$ degrees of freedom, with technical differences associated with the continuous phase-space framework.

The idea of studying anomalous weak values from the perspective of anomalous quasi-probabilities has also appeared before in a less general description. From a rather broad view, negative joint quasi-probabilities are always capable of reproducing experimental data in quantum theory~\cite{feynman1987negative,bartlett1945negative}, but there are many such distributions capable of reproducing the strongest possible quantum correlations~\cite{abramsky2014operational,morris2022witnessing,onggadinata2023reexamination,alSafi2013simulating}, a fact that somewhat disfavours those as good explanations due to fine-tuning arguments. Ref.~\cite{lund2010measuring} introduced the notion of ``weak value quasi-probability'', that relates to our description of $g$. Their distribution corresponds to the real part of the Kirkwood-Dirac quasi-probability distribution~\cite{kirkwood1933quantum,dirac1945analogy}, also known as Terletsky-Margeneau-Hill quasi-probabilities~\cite{levy2020quasiprobability,margenau1961correlation,terletsky1937limiting,dressel2015weak}. In our terms, it is the real part of the numerator defining $g$, with anomalous values then simply real values outside the range $[0,1]$.
Ref.~\cite{higgins2015using} later argued that negative values in this quasi-probability distribution can be used to single out the many possible joint distributions explaining quantum data, as this would be the only such distribution capable of explaining the results of weak measurements.

\section{Anomalous weak values require coherence}

An anomalous weak value $A_w$ is one that is outside of the range of eigenvalues of $A$. Clearly, the existence of an anomalous $A_w$ requires at least one anomalous quasi-probability~\cite{dressel2015weak}. Our first result is: 
\begin{theorem}\label{theorem: main result}
    The appearance of an anomalous weak value $A_w$ of observable $A$ requires coherence of both $\rho_\phi,\rho_\psi$ in the eigenbasis of $A$. 
\end{theorem}
\begin{proof}
First, let us show that a pair of (pure or mixed) states $\rho_{\psi}, \rho_{\phi}$ that are incoherent, that is, diagonal in the basis of $A$, cannot result in anomalous $A_w$. Let us assume $\rho_{\phi}, \rho_{\psi}$ are diagonal in $A$'s eigenbasis. As discussed in Ref.~\cite{wagner2023quantum}, for any set of states which are diagonal in a basis $A$, the invariants are the probability $p$ of getting the same outcome when measuring $A$ independently on all states. So the quasi-probabilities associated with $A_w$ can be rewritten as:
\begin{eqnarray}
&g(\rho_\phi,\rho_\psi\vert a_i)=\frac{\Delta_{\rho_{\phi} a_i \rho_{\psi}}} {\Delta_{\rho_{\phi} \rho_{\psi}}}=\nonumber\\
&\sum_{k=1}^d \frac{p(a_k, \rho_{\phi}) p(a_k,a_i ) p(a_k, \rho_{\psi})}{\Delta_{\rho_{\phi} \rho_{\psi}}}
=\frac{\Delta_{a_i \rho_{\phi}} \Delta_{a_i \rho_{\psi}}}{\Delta_{\rho_{\phi} \rho_{\psi}}}. \label{eq: quasiprob nonlinear invariant}
\end{eqnarray}

For diagonal, coherence-free states the quasi-probabilities $g(\rho_{\phi}, \rho_{\psi}|a_i)$ above define a genuine probability distribution: they are real non-negative values within $[0,1]$, and add up to 1. 
 In this case, $g(\rho_{\phi}, \rho_{\psi}|a_i)$ is the renormalized probability of obtaining equal outcomes $a_i$ after independent measurements of $A$ on $\ket{\psi}$, and $ \ket{\phi}$. As coherence-free states result in no anomalous quasi-probabilities, anomalous weak values are impossible for those states.

Let us now prove that anomalous weak values $A_w$ require that both $\rho_{\phi}$ and $\rho_{\psi}$ be coherent in $A$'s eigenbasis. Suppose, without loss of generality, that $\rho_{\phi}$ is diagonal in $A$'s basis, but $\rho_{\psi}$ is not. This implies that $\rho_\phi$ commutes with any $\vert a_i \rangle \langle a_i \vert$, and hence $\Delta_{\rho_\phi a_i \rho_\psi} = \text{Re}[\Delta_{\rho_\phi a_i \rho_\psi}] \leq 1$, the last inequality being a general feature of Bargmann invariants. Also, $\Delta_{\rho_{\phi} a_i \rho_{\psi}}=\Tr(\rho_{\phi} \ketbra{a_i}{a_i}\rho_{\psi}) = \Tr (\tau \rho_{\psi})$, where $\tau=\rho_{\phi}\ketbra{a_i}{a_i}$ is a  diagonal positive semi-definite matrix. 
It is known that for any two positive semidefinite matrices $X,Y$, the trace satisfies $\Tr(XY)\ge 0$ (see Theorem 1 of \cite{Yang95}). So, for the existence of negative or imaginary values of the invariants $\Delta_{\rho_{\phi} a_i \rho_{\psi}}$ we need both $\rho_{\psi}, \rho_{\phi}$ to be coherent. Positivity of all third order invariants implies that $g$ is non-anomalous, and therefore $A_w$ will also be non-anomalous. As we have seen, anomalous $A_w$ require at least one anomalous quasi-probability $g$. In case $g(\rho_\phi,\rho_\psi \vert a_i)>1$, we note that there must exist another $a_j \neq a_i$ such that $g(\rho_\phi,\rho_\psi\vert a_j)<0$ (due to the normalization of quasi-probabilities). As we have shown that negative values of $g$ are ruled out unless both $\rho_{\phi},\rho_{\psi}$ are coherent, this directly implies that values of $g$ larger than $1$ are also ruled out in this case.
\end{proof}

Some of the aspects outlined in the theorem above have appeared before. For instance, as mentioned in Ref.~\cite{dressel2010contextualvalues} negativity of $g$ arising from anomalous values of $A_w$ has been studied in the context of the so-called three-box problem in quantum foundations~\cite{aharonov2001threeboxproblem}, and in connections of anomalous weak values with Bell nonlocality~\cite{aharonov2015weak}. Ref.  \cite{wagner2023quantum} shows that negativity and imaginarity of the weak values $P_w^{(i)}$ of eigenprojectors of $A$ are witnesses of coherence, but not the general case of anomalous $A_w$, as in our Theorem \ref{theorem: main result}.

\begin{corollary}\label{corollary: main resul}
    Anomalous values of $g$ are sufficient for the existence of anomalous weak values for some observable, specifically, some eigenprojector of $A$'s. 
\end{corollary}
The proof is an immediate consequence of Theorem \ref{theorem: main result}:  any anomalous quasi-probability $g_i$ is an anomalous weak value $P_w^{(i)}$ for the associated eigenprojector $P^{(i)} \equiv |a_i\rangle\langle a_i|$ of $A$.

{
\section{Comparison with prior work}


Here we review some of the previously studied connections between quantum coherence and anomalous weak values, as a way to contextualize our results. In our view, these connections have been underappreciated up to now, as they are not mentioned in some comprehensive reviews on weak values \cite{dressel2014colloquium,tamir2013introduction}.

Ref.~\cite{mundarain2016quantumness} presented an analysis of how entanglement between the detection apparatus and the initial system is necessary for anomalous values that can be used in weak value amplification tasks. Their analysis is significantly less general than ours as it applies to an initial qubit system, discusses only real-valued weak values, relying on the usual weak measurement scheme, and using a much more intricate analysis of the Holevo quantity to conclude the necessity of quantum coherence. Our analysis, on the other hand, is broadly applicable, valid also to complex-valued weak values, and is not attached to any specific measurement scheme. 

Ref.~\cite{dressel2015weak} has a conceptual goal that is similar to ours: to show that anomaly arises from non-classical interference phenomena. The authors present a discussion of different methods for estimating weak values, and also briefly connect anomaly to negativity of quasiprobability distributions. No simple, general and formal argument such as Theorem~\ref{theorem: main result} is provided favoring an interpretation that weak values require quantum coherence.  Nevertheless, their results firmly establish the same conceptual result as ours, i.e., that any classical model reproducing the results of weak-measurements with anomalous weak-values must be capable of simulating properties of coherent quantum states. They make the connection between coherence and anomaly clear through negativity of quasiprobability, using a less general distribution than ours. For a comparison between our distribution $g$ and other constructions  previously proposed in the literature we refer to Sec.~\ref{sec: quasiprobability}. 

Finally, it is important to note that since anomalous weak values are proofs of quantum contextuality~\cite{pusey2014anomalous,kunjwal2019anomalous}, they are also proofs of the necessity of quantum coherence, as coherence is a necessary (yet not sufficient) condition for contextuality~\cite{wagner2022inequalities}. Still, the results from Refs.~\cite{pusey2014anomalous,kunjwal2019anomalous} heavily rely on the specific operational aspects of the weak measurement scheme, lacking in simplicity of the argument, specially if one is solely interested in coherence. We provide other comments on the connection with contextuality in Sec.~\ref{sec: contextuality}. 


}
\section{Coherence is not sufficient for weak value anomaly}

Theorem \ref{theorem: main result} above establishes that coherence is necessary for the appearance of anomalous weak values $A_w$. It is natural to ask whether it is also sufficient. In the following we show that in general coherence does not imply anomaly of $A_w$, and discuss particular conditions enabling results in this direction. We start with a simple example where anomalous quasi-probability values $g_i$ result in a non-anomalous weak value.

\textit{Example 1: anomalous quasi-probabilities yielding non-anomalous weak values.}-- Consider two rank-1 projectors $A=\ketbra{0}{0}, B=\ketbra{1}{1}$ in a 2-dimensional Hilbert space.  We can maximize negativity of $A_w$ with a configuration where $\ket{0}, \ket{\psi}, \ket{\phi}$ are separated by $120^o$ in a great circle of the Bloch sphere (see Fig. \ref{fig:awneg}). This results in a negative $A_w=-1/2$. The same choice of $\ket{\psi}, \ket{\phi}$ results in an anomalous weak value $B_w=3/2>1$. This example illustrates two points: 1)  anomalous weak values may arise from anomalous quasi-probabilities larger than 1, and not just from complex or negative values of the quasi-probabilities; and 2) even though both $A_w$ and $B_w$ are anomalous in this case, their sum gives a non-anomalous weak value for the identity operator $I_w=1$. 
As we can see, the fact that $A_w$ is an average weighted by quasi-probabilities means anomalous quasi-probabilities can average into a non-anomalous $A_w$.

\begin{figure}[h]
    
    \includegraphics[width=0.95\columnwidth]{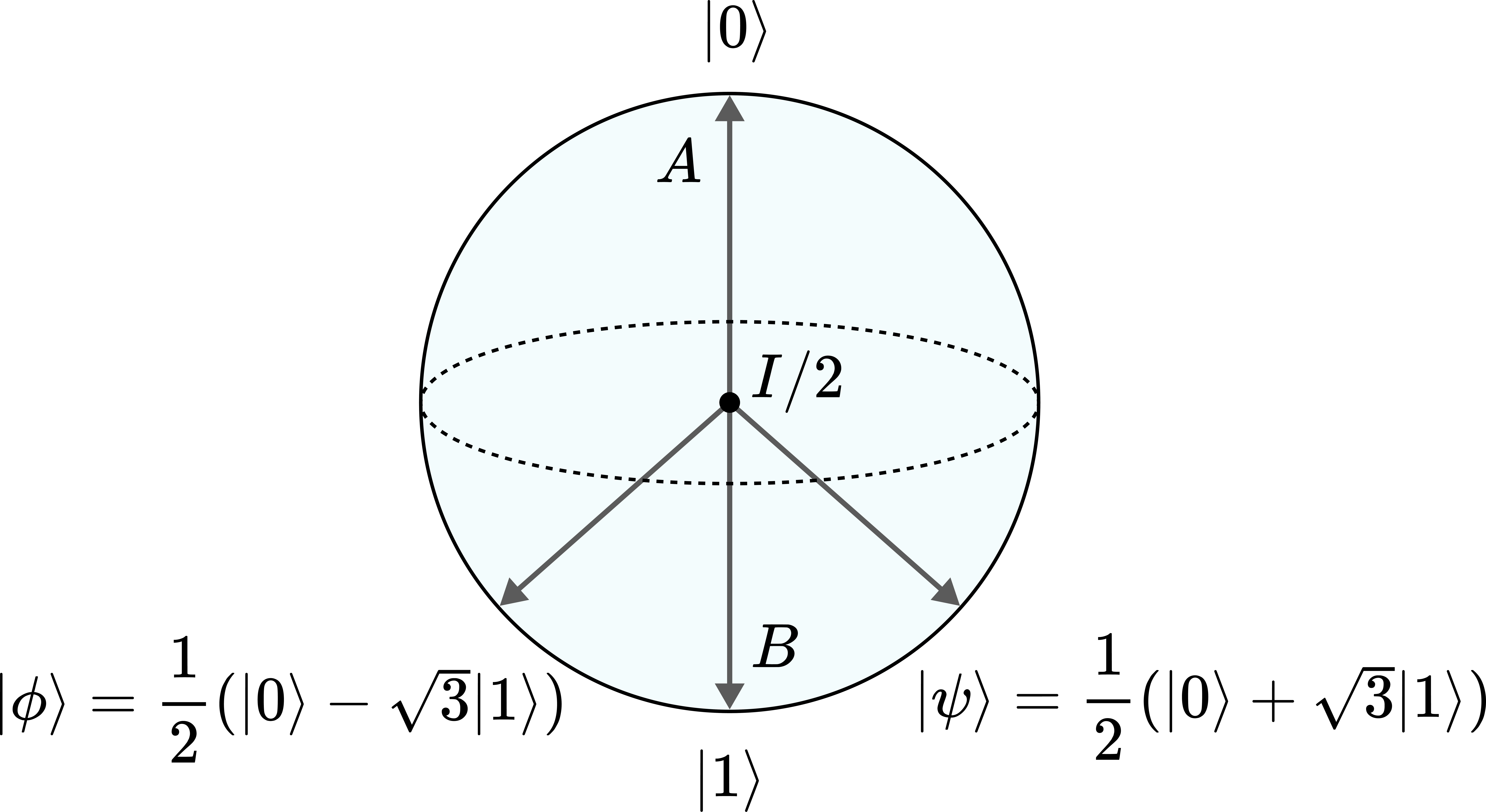}
    \centering
    \caption{\textbf{Example of anomalous weak values.} The weak value $A_w$ for the projector $A=\ketbra{0}{0}$, with $\ket{\phi}, \ket{\psi}$ chosen as in the figure, results in the anomalous $A_w=\frac{1}{1/4}\langle \phi|0\rangle \langle 0|\psi\rangle \langle \psi|\phi \rangle = \frac{1}{1/4} \cdot \left(\frac{1}{2}\right) \cdot\left(\frac{1}{2}\right)\cdot\left(-\frac{1}{2}\right)=-\frac{1}{2}<0.$ A similar calculation gives anomalous $B_w=3/2>1$. The weak value of the identity operator is non-anomalous: $I_w=A_w+B_w=1$.}
    
    \label{fig:awneg}
\end{figure}

As discussed in Refs.~\cite{chien2016characterization, oszmaniec2021measuring}, all unitary-invariant quantities of a set of states can be written in terms of Bargmann invariants. As we have seen in the proof of Theorem \ref{theorem: main result}, each quasi-probability  $g(\rho_\phi,\rho_\psi\vert a_i)$ is a function of such unitary-invariant quantities, which in the case of coherence-free states must be non-anomalous, i.e. real in the range $[0,1]$. A natural question is whether any pair of coherent states $\rho_{\phi}, \rho_{\psi}$ leads to anomalous values for $g_i$. This would signal that coherence is sufficient for the appearance of anomalous values of some observable. 

As it turns out, various sets of coherent states lead to non-anomalous distributions $g$, independently of whether the pre- and post-selected states are pure or mixed.  

\textit{Example 2: coherent states yielding non-anomalous quasi-probabilities.}-- Let $\rho_\phi,\rho_\psi \in \mathcal{D}(\mathcal{H})$ be an arbitrary pair of non-commuting density matrices: $[\rho_\phi,\rho_\psi] \neq 0$. Non-commutativity guarantees coherence with respect to any basis, hence in particular $A$'s eigenbasis. Consider the corresponding real-amplitude states, by mapping $\rho_\phi \mapsto \rho_\phi^{\mathbb{R}} = (\rho_\phi + \rho_\phi^T)/2$ and similarly for $\rho_\psi$. Random generation of state pairs will rapidly turn up examples with only non-anomalous  $g(\rho_\phi^{\mathbb{R}},\rho_\psi^{\mathbb{R}} | a_i) \in [0,1], \forall i$. Here is a qubit example:
\begin{equation*}
    \rho_\psi = \left(\begin{matrix}
    \frac{3}{4} & \sqrt{\frac{3}{32}}\\ \sqrt{\frac{3}{32}} & \frac{1}{4}\end{matrix}\right), \,\rho_\phi = \left(\begin{matrix}
        \frac{3}{4} & \frac{\sqrt{3}}{8} \\ \frac{\sqrt{3}}{8} & \frac{1}{4}
    \end{matrix}\right)
\end{equation*}
for which $g_0 = 0.829997$ and $g_1 = 0.170003$.

\section{Anomaly of $A_w$ as a witness of generalized contextuality}\label{sec: contextuality}

Of particular significance for the discussion of nonclassicality of weak values is their connection with 
generalized contextuality~\cite{spekkens2005contextuality}. Noncontextual models can reproduce some aspects of quantum superpositions~\cite{catani2021why,spekkens2007evidence}, despite being arguably classical from many viewpoints~\cite{bartlet2012reconstruction,baldi2021emergence,shahandeh2021contextuality,schmid2021characterization,spekkens2019ontological}. 
Refs.~\cite{pusey2014anomalous,kunjwal2019anomalous} show that noncontextual models cannot explain the data arising from weak measurements. As weak values can be measured in other ways \cite{wagner2023quantum}, we wonder if it is possible to obtain no-go results such as those in \cite{pusey2014anomalous,kunjwal2019anomalous}, \textit{without relying} on specifics of the operational weak measurement set-up. In light of our results, a simple way to do so is to use the event graph approach~\cite{wagner2022inequalities}, by studying the graph of Fig.~\ref{fig:awgraph} where one only imposes constraints over edges/two-state overlaps (as opposed to cycles in the frame graph representing higher-order invariants). Given a graph $G$, one defines polytopes $C_G$ whose facets can be translated into noncontextuality inequalities~\cite{wagner2022inequalities}. For the graph of Fig.~\ref{fig:awgraph} the only non-trivial inequalities will be $3$-cycle inequalities, 
\begin{equation}\label{eq: 3cycle inequality}
    h_3 := \Delta_{\phi\psi}+\Delta_{\phi a_i}-\Delta_{\psi a_i} \leq 1,
\end{equation}
and sign permutations.

Consider the simplest case, which corresponds to the graph associated with dimension $d=2$.
Again, we will use the quasi-probability $g_i$ to establish a connection of anomaly with a notion of non-classicality, in this case generalized contextuality.
For qubits with only real-valued amplitudes, whenever we have an anomalous value $g_i>1$ this implies, from the results of Ref.~\cite{wagner2023quantum}, that some $h_3 > 1$, an example of inequality violation.  
This violation can be used, together with the results from Refs.~\cite{wagner2022inequalities,wagner2022coherence} to construct prepare-and-measure fragments of quantum theory~\cite{selby2021accessible}, defined by a pair $(\mathcal{S},\mathcal{E})$ of sets of states and sets of measurement effects. If $g(\phi,\psi|a_i)>1$ holds, then letting $\mathcal{S} = \{\vert\phi\rangle,\vert \psi \rangle,\vert a_1\rangle,\vert a_2\rangle ,\vert\phi^\perp\rangle,\vert \psi^\perp \rangle\} $, where $\vert \phi^\perp \rangle $ is the antipodal state of $\vert \phi \rangle$ in the Bloch sphere, and effects $\mathcal{E} = \mathcal{S}$, then it can be shown that the fragment $(\mathcal{S},\mathcal{E})$ cannot have a noncontextual explanation. To robustly test this, one can use linear programming techniques that will indicate presence of contextuality directly, and moreover return robustness to depolarizing~\cite{selby2022open} and dephasing noise~\cite{rossi2022contextuality}. The case for $d>2$, or anomalous imaginary values of $g$ are not so direct but can also be analysed with the tools discussed here. 
This is an alternative argument that anomalous weak values imply generalized contextuality, different from the approach of Refs. \cite{pusey2014anomalous,kunjwal2019anomalous}.

\section{Discussion and further directions}

We have characterized anomalous weak values in terms of a complex-valued quasi-probability distribution. Although this distribution has appeared before, our treatment can be viewed either as a generalization of other constructions for real-valued quantities, or as a discrete version of a continuous-variable analogue. With this tool, we show that  coherence is necessary for anomalous weak values to occur; moreover, we show that specifically imaginarity or real values for $g$ outside the range $[0,1]$ are also sufficient for the appearance of anomalous weak values. We also presented examples showing that coherence does not necessarily yield anomalous weak values. 

Our results have applications in terms of simplifying tests of generalized contextuality associated with anomalous weak values, via direct quantum circuit measurements of weak values that do not appeal to specific operational aspects of weak measurement schemes.\\

\begin{acknowledgements}
We thank Eliahu Cohen, Amit Te'eni, Rui Soares Barbosa and Ismael L. Paiva for helpful discussions and feedback on an early version of this manuscript. We acknowledge financial support from FCT -- Fundação para a Ciência e a Tecnologia (Portugal) via PhD Grant SFRH/BD/151199/2021 (RW) and via Grant CEECINST/00062/2018 (EFG). This work was supported by the Horizon Europe project FoQaCiA, GA no. 101070558.
\end{acknowledgements}

\bibliography{bibliography}

\end{document}